\documentclass[runningheads]{llncs}

\usepackage{xcolor}
\definecolor{TUMblue}{RGB}{0,101,189}
\definecolor{TUMbluedark}{RGB}{0,82,147}
\definecolor{TUMbluebright}{RGB}{152,198,234}
\definecolor{TUMbluemoderate}{RGB}{100,160,200}
\definecolor{TUMivory}{RGB}{218,215,203}
\definecolor{TUMgreen}{RGB}{90,160,20}
\definecolor{TUMorange}{RGB}{227,114,34}
\definecolor{TUMgray}{gray}{0.6}

\usepackage{amssymb}
\usepackage{amsmath}
\usepackage{amstext}
\usepackage{amsbsy}
\usepackage{amsfonts}
\usepackage{dsfont}
\usepackage{marvosym}
						
\usepackage[noadjust]{cite} 

\usepackage{enumitem}

\usepackage[hidelinks]{hyperref}

\usepackage{graphicx}
\usepackage[section]{placeins}

\usepackage{tabularx}
\usepackage{tabulary}
\usepackage{booktabs}
\usepackage{array}
\usepackage{multirow}
\usepackage{mathdots}
\usepackage{chngcntr}

\usepackage{tikz}
\usepackage{tkz-euclide}
\usetikzlibrary{positioning,calc,arrows,shapes,backgrounds}

\usepackage{authblk}

\usepackage{glossaries}
\glsdisablehyper

\usepackage[textsize=scriptsize]{todonotes} 

\usepackage{pdfpages}

\usepackage[ruled,linesnumbered]{algorithm2e}
\SetKwInOut{Input}{Input}
\SetKwInOut{Output}{Output}
\SetKwInOut{Set}{Set}

\newcommand{\st}{\text{s.t.\ }}
\newcommand{\wloss}{\text{w.l.o.g.\ }}
\newcommand{\etal}{\text{et al.\ }}
\newcommand{\eps}{\varepsilon}
\newcommand{\R}{\mathds{R}}

\newcommand{\ubar}{\bar{u}}

\newcommand{\pjA}{p_j^A}
\newcommand{\pkA}{p_k^A}
\DeclareMathOperator*{\argmin}{argmin}

\DeclareMathOperator*{\ALG}{ALG}
\DeclareMathOperator*{\OPT}{OPT}
\title{Explorable Uncertainty in Scheduling with Non-Uniform Testing Times\thanks{Funded by the Deutsche Forschungsgemeinschaft (DFG, German Research Foundation) – 277991500/GRK2201, and by the European Research Council, Grant Agreement No.\ 691672.}}
\author{Susanne Albers\inst{1}\and Alexander Eckl\inst{1,2,}\thanks{Corresponding author, eMail: alexander.eckl@tum.de}}
\authorrunning{S. Albers and A. Eckl}
\institute{Department of Informatics, Technical University of Munich,\\ Boltzmannstr. 3, 85748 Garching, Germany\\ \email{albers@in.tum.de, alexander.eckl@tum.de}\\
\and
Advanced Optimization in a Networked Economy, Technical University of Munich,\\ Arcisstra\ss e 21, 80333 Munich, Germany}

\begin{document}
\maketitle
\begin{abstract}
	The problem of scheduling with testing in the framework of explorable uncertainty models environments where some preliminary action can influence the duration of a task. In the model, each job has an unknown processing time that can be revealed by running a test. Alternatively, jobs may be run untested for the duration of a given upper limit. Recently, D\"urr \etal \cite{Duerr2018} have studied the setting where all testing times are of unit size and have given lower and upper bounds for the objectives of minimizing the sum of completion times and the makespan on a single machine. In this paper, we extend the problem to non-uniform testing times and present the first competitive algorithms. The general setting is motivated for example by online user surveys for market prediction or querying centralized databases in distributed computing. Introducing general testing times gives the problem a new flavor and requires updated methods with new techniques in the analysis. We present constant competitive ratios for the objective of minimizing the sum of completion times in the deterministic case, both in the non-preemptive and preemptive setting. For the preemptive setting, we additionally give a first lower bound. We also present a randomized algorithm with improved competitive ratio. Furthermore, we give tight competitive ratios for the objective of minimizing the makespan, both in the deterministic and the randomized setting.
\end{abstract}

\section{Introduction}

In scheduling environments, uncertainty is a common consideration for optimization problems. Commonly, results are either based on worst case considerations or a random distribution over the input. These approaches are known as robust optimization and stochastic optimization, respectively. However, it is often the case that unknown information can be attained through investing some additional resources, e.g.\ time, computing power or money. In his seminal paper, Kahan \cite{Kahan1991} has first introduced the notion of explorable or queryable uncertainty to model obtaining additional information for a problem at a given cost during the runtime of an algorithm. Since then, these kind of problems have been explored in different optimization contexts, for example in the framework of combinatorial, geometric or function value optimization tasks. 

Recently, D\"urr \etal \cite{Duerr2018} have introduced a model for scheduling with testing on a single machine within the framework of explorable uncertainty. In their approach, a number of jobs with unknown processing times are given. Testing takes one unit of time and reveals the processing time. If a job is executed untested, the time it takes to run the job is given by an upper bound. The novelty of their approach lies in having tests executed directly on the machine running the jobs as opposed to considering tests separately. 

In view of this model, a natural extension is to consider non-uniform testing times to allow for a wider range of problems. D\"urr \etal state that for certain applications it is appropriate to consider a broader variation on testing times and leave this question up for future research.

Situations where a preliminary action, operation or test can be executed before a job are manifold and include a wide range of real-life applications. In the following, we discuss a small selection of such problems and emphasize cases with heterogeneous testing requirements. Consider first a situation where an online user survey can help predict market demand and production times. The time needed to produce the necessary amount of goods for the given demand is only known after conducting the survey. Depending on its scope and size, the invested costs for the survey may vary significantly.

As a second example, we look at distributed computing in a setting with many distributed local databases and one centralized master server. At the local stations, only estimates of some data values are stored; in order to obtain the true value one must query the master server. It depends on the distance and connection quality from any localized database to the master how much time and resources this requires. Olston and Widom \cite{Olston2000} have considered this setting in detail.

Another possible example is the acquisition of a house through an agent giving us more information about its value, location, condition, etc., but demanding a price for her services. This payment could vary based on the price of the house, the amount of work of the agent or the number of competitors.

In their paper, D\"urr \etal \cite{Duerr2018} mention fault diagnosis in maintenance and medical treatment, file compression for transmissions, and running jobs in an alternative fast mode whose availability can be determined through a test. Generally, any situation involving diverse cost and duration estimates, like e.g.\ in construction work, manufacturing or insurance, falls into our category of possible applications.

In view of all these examples, we investigate non-uniform testing in the scope of explorable uncertainty on a single machine as introduced by \cite{Duerr2018}. We study whether algorithms can be extended to this non-uniform case and if not, how we can find new methods for it.

\subsection{Problem Statement}

We consider $n$ jobs to be scheduled on a single machine. Every job $j$ has an unknown processing time $p_j$ and a known upper bound $u_j$. It holds $0 \le p_j \le u_j$ for all $j$. Each job also has a testing time $t_j \ge 0$. A job can either be executed untested, which takes time $u_j$, or be tested and then executed, which takes a total time of $t_j + p_j$. Note that a tested job does not necessarily have to be executed right after its test, it may be delayed arbitrarily while the algorithm tests or executes other jobs.

Since only the upper bounds are initially known to the algorithm, the task can be viewed as an online problem with an adaptive adversary. The actual processing times $p_j$ are only realized after job $j$ has been tested by the algorithm. In the randomized case, the adversary knows the distribution of the random input parameters of an algorithm, but not their outcome.

We denote the completion time of a job $j$ as $C_j$ and primarily consider the objective of minimizing the total sum of completion times $\sum_j C_j$. As a secondary objective, we also investigate the simpler goal of minimizing the makespan $\max_j C_j$. We use competitive analysis to compare the value produced by an algorithm with an optimal offline solution.

Clearly, in the offline setting where all processing times are known, an optimal schedule can be determined directly: If $t_j + p_j \le u_j$ then job $j$ is tested, otherwise it is run untested. For the sum of completion times, the jobs are therefore scheduled in order of non-decreasing $\min(t_j+p_j,u_j)$. Any algorithm for the online problem not only has to decide whether to test a given job or not, but also in which order to run all tests and executions of both untested and tested jobs. For a solution to the makespan objective, the ordering of the jobs does not matter and an optimal offline algorithm decides the testing by the same principle as above.

\subsection{Related Work}

Our setting is directly based on the problem of scheduling uncertain jobs on a single machine with explorable processing times, introduced by D\"urr \etal \cite{Duerr2018} in 2018. They only consider the special case where $t_j \equiv 1$ for all jobs. For deterministic algorithms, they give a lower bound of 1.8546 and an upper bound of 2. In the randomized case, they give a lower bound of 1.6257 and a 1.7453-competitive algorithm. For several deterministic special case instances, they provide upper bounds closer to the best possible ratio of 1.8546. Additionally, tight algorithms for the objective of minimizing the makespan are given for both the deterministic and randomized cases.

Testing and executing jobs on a single machine can be viewed as part of the research area of \emph{queryable uncertainty} or \emph{explorable uncertainty}. The first seminal paper on dealing with uncertainty by querying parts of the input was published in 1991 by Kahan \cite{Kahan1991}. In his paper, Kahan considers a set of elements with uncertain values that lie in a closed interval. He explores approximation guarantees for the number of queries necessary to obtain the maximum and median value of the uncertain elements.

Since then, there has been a large amount of research concerned with the objective of minimizing the number of queries to obtain a solution. A variety of numerical, geometric and combinatorial problems have been studied in this framework, the following is a selection of some of these publications: Next to Kahan, Feder \etal \cite{Feder2003}, Khanna and Tan \cite{Khanna2001}, and Gupta \etal \cite{Gupta2011} have also considered the objective of determining different function values, in particular the k-smallest value and the median. Bruce \etal \cite{Bruce2005} have analysed geometric tasks, specifically the Maximal Points and Convex Hull problems. They have also introduced the notion of \emph{witness sets} as a general concept for queryable uncertainty, which was then generalized by Erlebach \etal \cite{Erlebach2008}. Olston and Widom \cite{Olston2000} researched caching problems while allowing for some inaccuracy in the objective function. Other studied combinatorial problems include minimum spanning tree \cite{Erlebach2008,Megow2017}, shortest path \cite{Feder2007}, knapsack \cite{Goerigk2015} and boolean trees \cite{Charikar2002}. See also the survey by Erlebach and Hoffmann \cite{Erlebach2015} for an overview over research in this area.

A related type of problems within optimization under uncertainty are settings where the cost of the queries is a direct part of the objective function. Most notably, the paper by D\"urr \etal \cite{Duerr2018} falls into this category. There, the tests necessary to obtain additional information about the runtime of the jobs are executed on the same machine as the jobs themselves. Other examples include Weitzman's original Pandora's Box problem \cite{Weitzman1979}, where $n$ independent random variables are probed to maximize the highest revealed value. Every probing incurs a price directly subtracted from the objective function. Recently, Singla \cite{Singla2018} introduced the 'price of information' model to describe receiving information in exchange for a probing price. He gives approximation ratios for various well-known combinatorial problems with stochastic uncertainty. 

%
\subsection{Contribution}

In this paper, we provide the first algorithms for the more general scheduling with testing problem where testing times can be non-uniform. Consult Table \ref{tab:results_expl_uncer} for an overview of results for both the non-uniform and uniform versions of the problem. All ratios provided without citation are introduced in this paper. The remaining results are presented in \cite{Duerr2018}.

\begin{table}[htb]
	\centering
	\caption{Overview of results}
	\label{tab:results_expl_uncer}
	\begin{tabulary}{\textwidth}{C | C | C | C}
		\textbf{Objective Type} 		& \textbf{General tests}	& \textbf{Uniform tests}			& \textbf{Lower bound}				\\ \midrule
		$\sum C_j$ - deterministic		& $4$						& $2$ \cite{Duerr2018}				& $1.8546$ \cite{Duerr2018}			\\
		$\sum C_j$ - randomized			& $3.3794$					& $1.7453$ \cite{Duerr2018}			& $1.6257$ \cite{Duerr2018}			\\
		$\sum C_j$ - determ. preemptive	& $2 \varphi \approx 3.2361$& -									& $1.8546$							\\
		$\max C_j$ - deterministic		& $\varphi \approx 1.6180$	& $\varphi$	\cite{Duerr2018}		& $\varphi$ \cite{Duerr2018}		\\
		$\max C_j$ - randomized			& $\frac{4}{3}$				& $\frac{4}{3}$ \cite{Duerr2018}	& $\frac{4}{3}$ \cite{Duerr2018}	\\
	\end{tabulary}
\end{table}

For the problem of scheduling uncertain jobs with non-uniform testing times on a single machine, our results are the following: A deterministic 4-competitive algorithm for the objective of minimizing the sum of completion times and a randomized 3.3794-competitive algorithm for the same objective. If we allow preemption - that is, to cancel the execution of a job at any time and start working on a different job - then we can improve the deterministic case to be $2\varphi$-competitive. Here, $\varphi \approx 1.6180$ is the golden ratio.

For the objective of minimizing the makespan, we adopt and extend the ideas of D\"urr \etal \cite{Duerr2018} to provide a tight $\varphi$-competitive algorithm in the deterministic case and a tight $\frac{4}{3}$-competitive algorithm in the randomized case.

Our approaches handle non-uniform testing times in a novel fashion distinct from the methods of \cite{Duerr2018}. As we show in Appendix \ref{append:small_upper_limits}, the idea of scheduling untested jobs with small upper bounds in the beginning of the schedule, which works well in the uniform case, fails to generalize to non-uniform tests. Additionally, describing parameterized worst-case instances becomes intangible in the presence of an arbitrary number of different testing times.

In place of these methods, we compute job completion times by cross-e\-xa\-min\-ing contributions of other jobs in the schedule. We determine tests based on the ratio between the upper bound and the given test time and pay specific attention to sorting the involved executions and tests in an suitable way.

The paper is structured as follows: Sections \ref{sec:deterministic} and \ref{sec:randomized} examine the deterministic and randomized cases respectively. Various algorithms are presented and their competitive ratios proven. We extend the optimal results for the objective of minimizing the makespan from the uniform case to general testing times in Section \ref{sec:makespan}. Finally, we conclude with some open problems.

\section{Deterministic Setting}
\label{sec:deterministic}

In this section, we introduce our basic algorithm and prove deterministic upper bounds for the non-preemptive as well as the preemptive case. The basic structure introduced in Section \ref{subsec:basic_alg_and_proof} works as a framework for other algorithms presented later. We give a detailed analysis of the deterministic algorithm and prove that it is $4$-competitive if parameters are chosen accordingly. In Section \ref{subsec:det_preemption} we prove that an algorithm for the preemptive case is 3.2361-competitive and that no preemptive algorithm can have a ratio better than $1.8546$.

\subsection{Basic Algorithm and Proof of 4-Competitiveness}
\label{subsec:basic_alg_and_proof}

We now present the elemental framework of our algorithm, which we call \emph{$(\alpha,\beta)$-SORT}. As input, the algorithm has two real parameters, $\alpha \ge 1$ and $\beta \ge 1$.
\begin{algorithm}[ht]
	$T \leftarrow \emptyset$, $N \leftarrow \emptyset$, $\sigma_j \equiv 0$\;
	\ForEach{$j\in[m]$}{
		\eIf{$u_j \ge \alpha t_j$}
			{add $j$ to $T$\; set $\sigma_j \leftarrow \beta t_j$\;}
			{add $j$ to $N$\; set $\sigma_j \leftarrow u_j$\;}
	}
	\While{$N \cup T \neq \emptyset$}{
		choose $j_{\min} \in \argmin_{j \in N \cup T} \sigma_j$\;
		\uIf{$j_{\min} \in N$}{execute $j_{\min}$ untested\; remove $j_{\min}$ from $N$\;}
		\uElseIf{$j_{\min} \in T$}{
			\eIf{$j_{\min}$ not tested}
				{test $j_{\min}$\; set $\sigma_{j_{\min}} \leftarrow p_{j_{\min}}$\;}
				{execute $j_{\min}$\; remove $j_{\min}$ from $T$\;}
		}
	}
	\caption{$(\alpha,\beta)$-SORT}
	\label{alg:alpha_beta_sort}
\end{algorithm}

The algorithm is divided into two phases. First, we decide for each job whether we test this job or not based on the ratio $\frac{u_j}{t_j}$. This gives us a partition of $[m]$ into the disjoint sets ${T = \{j\ \in [m]: \text{ALG tests } j\}}$ and $N = \{j\ \in [m]: \text{ALG runs } j \text{ untested}\}$. In the second phase, we always attend to the job $j_{\min}$ with the current smallest \emph{scaling time} $\sigma_j$. The scaling time is the time needed for the next step of executing $j$: 
\begin{itemize}
	\item If $j$ is in $N$, then $\sigma_j = u_j$.
	\item If $j$ is in $T$ and has not been tested, then $\sigma_j = \beta t_j$.
	\item If $j$ is in $T$ and has already been tested, then $\sigma_j = p_j$.
\end{itemize}

Note that in the second case above, we 'stretch' the scaling time by multiplying with ${\beta \ge 1}$. The intention behind this stretching is that testing a job, unlike executing it, does not immediately lead to a job being completed. Therefore the parameter $\beta$ artificially lowers the relevance of testing in the ordering of our algorithm. Note that the actual time needed for testing remains $t_j$.

In the following, we show that the above algorithm achieves a provably good competitive ratio. The parameters are kept general in the proof and are then optimized in a final step. We present the computations with general parameters for a clearer picture of the proof structure, which we will reuse in later sections. In the final optimization step it will turn out that setting $\alpha = \beta = 1$ yields a best-possible competitive ratio of $4$.

\begin{theorem}
	\label{thm:1_1_SORT}
	The $(1,1)$-SORT algorithm is $4$-competitive for the objective of minimizing the sum of completion times.
\end{theorem}

\begin{proof}
For the purpose of estimating the algorithmic result against the optimum, let $\rho_j := {\min(u_j,t_j + p_j)}$ be the optimal running time of job $j$. Without loss of generality, we order the jobs \st ${\rho_1 \ge \hdots \ge \rho_n}$. Hence the objective value of the optimum is
\begin{equation}
\label{eq:opt_value}
\OPT = \sum_{j=1}^n j \cdot \rho_j
\end{equation}
Additionally, let
\begin{equation}
\label{eq:alg_proc_time}
	\pjA := \begin{cases}
		t_j + p_j 	& \text{if } j \in T,\\
		u_j 		& \text{if } j \in N,
	\end{cases}
\end{equation}
be the \emph{algorithmic running time} of $j$, i.e.\ the time the algorithm spends on running job $j$.

We start our analysis by comparing $\pjA$ to the optimal runtime $\rho_j$ for a single job, summarized in the following Proposition:
\begin{proposition}
\label{prop:claims_case_dist}
	\begin{enumerate}[label=\emph{(\alph*)}]
		\item $\forall j \in T$: $t_j \le \rho_j$, $p_j \le \rho_j$
		\item $\forall j \in T$: $\pjA \le \left(1+\frac{1}{\alpha}\right) \rho_j$
		\item $\forall j \in N$: $\pjA \le \alpha \rho_j$
	\end{enumerate}
\end{proposition}

Part (a) directly estimates testing and running times of tested jobs against the values of the optimum. We will use this extensively when computing the completion time of the jobs.
The proof of parts (b) and (c) is very similar to the proof of Theorem 14 in \cite{Duerr2018} for uniform testing times. We refer to the appendix for a complete write-down of the proof. Note that instead of considering a single bound, we split the upper bound of the algorithmic running time $\pjA$ into different results for tested (b) and untested jobs (c). This allows us to differentiate between different cases in the proof of Lemma \ref{lem:contribution_lemma} in more detail. We will often make use of this Proposition to upper bound the algorithmic running time in later sections.

To obtain an estimate of the completion time $C_j$, we consider the \emph{contribution} $c(k,j)$ of all jobs $k \in [n]$ to $C_j$. We define $c(k,j)$ to be the amount of time the algorithm spends scheduling job $k$ before the completion of $j$. Obviously it holds that $c(k,j) \le \pkA$. The following central lemma computes an improved upper bound on the contribution $c(k,j)$, using a rigorous case distinction over all possible configurations of $k$ and $j$:

\begin{lemma}[Contribution Lemma]
\label{lem:contribution_lemma}
	Let $j \in [n]$ be a given job. The completion time of $j$ can be written as
	\begin{equation*}
	C_j = \sum_{k \in [n]} c(k,j).
	\end{equation*}
	Additionally, for the contribution of $k$ to $j$ it holds that
	\begin{equation*}
		c(k,j) \le \max\left(\left(1+\frac{1}{\beta}\right)\alpha,1+\frac{1}{\alpha},1+\beta\right) \rho_j.
	\end{equation*}
\end{lemma}

Refer to Appendix \ref{append:contribution_lemma} for the proof. Depending on whether $j$ and $k$ are tested or not, the lemma computes various upper bounds on the contribution using estimates from Proposition \ref{prop:claims_case_dist}. Finally, the given bound on $c(k,j)$ is achieved by taking the maximum over the different cases.

Recall that the jobs are ordered by non-increasing optimal execution times $\rho_j$, which by Proposition \ref{prop:claims_case_dist} are directly tied to the algorithmic running times. Hence, the jobs $k$ with small indices are the 'bad' jobs with possibly large running times. For jobs with $k \le j$ we therefore use the independent upper bound from the Contribution Lemma. Jobs with large indices $k>j$ are handled separately and we directly estimate them using their running time $\pkA$.

By Lemma \ref{lem:contribution_lemma} and Proposition \ref{prop:claims_case_dist}(b),(c) we have
\begin{equation*}
	\begin{aligned}
		C_j 	&= \sum_{k>j} c(k,j) + \sum_{k \le j} c(k,j)\\
				&\le \sum_{k>j} \pkA + \sum_{k \le j} \max\left(\left(1+\frac{1}{\beta}\right)\alpha,1+\frac{1}{\alpha},1+\beta\right) \rho_j\\
				&= \sum_{k>j} \max\left(\alpha,1+\frac{1}{\alpha}\right) \rho_k + \max\left(\left(1+\frac{1}{\beta}\right)\alpha,1+\frac{1}{\alpha},1+\beta\right) j \cdot \rho_j.\\
	\end{aligned}
\end{equation*}

Finally, we sum over all jobs $j$:
\begin{equation*}
	\begin{aligned}
		\sum_{j=1}^n C_j 	&= \sum_{j=1}^n \sum_{k=j+1}^n \max\left(\alpha,1+\frac{1}{\alpha}\right) \rho_k\\
							&\qquad + \sum_{j=1}^n \max\left(\left(1+\frac{1}{\beta}\right)\alpha,1+\frac{1}{\alpha},1+\beta\right) j \cdot \rho_j\\
							&= \max\left(\alpha,1+\frac{1}{\alpha}\right) \sum_{j=1}^n (j-1) \rho_j\\
							&\qquad + \max\left(\left(1+\frac{1}{\beta}\right)\alpha,1+\frac{1}{\alpha},1+\beta\right) \sum_{j=1}^n j \cdot \rho_j\\
							&\le \underbrace{ \left( \max\left(\alpha,1+\frac{1}{\alpha}\right) + \max\left(\left(1+\frac{1}{\beta}\right)\alpha,1+\frac{1}{\alpha},1+\beta\right) \right) }_{=: f(\alpha,\beta)} \sum_{j=1}^n j \cdot \rho_j\\
							&= f(\alpha,\beta) \cdot \OPT\\
	\end{aligned}
\end{equation*}
Minimizing $f(\alpha,\beta)$ on the domain $\alpha,\beta \ge 1$ yields optimal parameters $\alpha = \beta = 1$ and a value of $f(1,1) = 4$. We conclude that $(1,1)$-SORT is $4$-competitive.
\end{proof}

The parameter selection $\alpha = 1$, $\beta = 1$ is optimal for the closed upper bound formula we obtained in our proof. It is possible and somewhat likely that a different parameter choice leads to better overall results for the algorithm. In the optimal makespan algorithm (see Section \ref{sec:makespan}) the value of $\alpha$ is higher, suggesting that $\alpha = 1$, which leads to testing all non-trivial jobs, might not be the best choice. The problem structure and the approach by D\"urr \etal \cite{Duerr2018} also motivate setting $\beta$ to some higher value than $1$. For our proof, setting parameters like we did is optimal.

In the appendix, we take advantage of this somewhat unexpected parameter outcome to prove that $(1,1)$-SORT cannot be better than $3$-competitive. Additionally, we show that for \emph{any} choice of parameters, $(\alpha,\beta)$-SORT is not better than 2-competitive.

\subsection{A Deterministic Algorithm with Preemption}
\label{subsec:det_preemption}

The goal of this section is to show that if we allow jobs to be preempted there exists a $3.2361$-competitive algorithm. In his book on Scheduling, Pinedo \cite{Pinedo2016} defines preemption as follows: "The scheduler is allowed to interrupt the processing of a job (preempt) at any point in time and put a different job on the machine instead."

The idea for our algorithm in the preemptive setting is based on the so-called \emph{Round Robin} rule, which is used frequently in preemptive machine scheduling \cite[Chapters 3.7, 5.6, 12.4]{Pinedo2016}. The scheduling time frame is divided into very small equal-sized units. The Round Robin algorithm then cycles through all jobs, tending to each job for exactly one unit of time before switching to the next. It ensures that at any time the amount every job has been processed only differs by at most one time unit \cite{Pinedo2016}.

The Round Robin algorithm is typically applied when job processing times are completely unknown. In our setting, we are actually given some upper bounds for our processing times and may invest testing time to find out the actual values. Despite having more information, it turns out that treating all job processing times as unknown in a Round Robin setting gives a provably good result. The only way we employ upper bounds and testing times is again to decide which jobs will be tested and which will not. We again do this at the beginning of our schedule for all given jobs. The rule to decide testing is exactly the same as in the first phase of Algorithm \ref{alg:alpha_beta_sort}: If $u_j/t_j \ge \alpha$, then test $j$, otherwise run $j$ untested. Again, $\alpha$ is a parameter that is to be determined. It will turn out that setting $\alpha = \varphi$ gives the best result.

The pseudo-code for the \emph{Golden Round Robin} algorithm is given in Algorithm~\ref{alg:golden-round-robin}.

\begin{algorithm}[ht]
	$T \leftarrow \emptyset$, $N \leftarrow \emptyset$, $\sigma_j \equiv 0$\;
	\ForEach{$j\in[m]$}{
		\eIf{$u_j \ge \varphi t_j$}
		{add $j$ to $T$\; set $\sigma_j \leftarrow t_j$\;}
		{add $j$ to $N$\; set $\sigma_j \leftarrow u_j$\;}
	}
	\While{$\exists j\in[m]$ not completely scheduled}{
		run Round Robin on all jobs using $\sigma_j$ as their processing time\;
		let $j_{\min}$ be the first job to finish during the current execution\;
		\If{$j_{\min} \in T$ and $j_{\min}$ tested but not executed}{set $\sigma_{j_{\min}} \leftarrow p_{j_{\min}}$ and keep $j_{\min}$ in the Round Robin rotation\;}
	}
	\caption{Golden Round Robin}
	\label{alg:golden-round-robin}
\end{algorithm}

Essentially, the algorithm first decides for all jobs whether to test them and then runs a regular Round Robin scheme on the algorithmic testing time $\pjA$, which is defined as in \eqref{eq:alg_proc_time}. 

\begin{theorem}
\label{thm:det_preempt}
	The Golden Round Robin algorithm is $3.2361$-competitive in the preemptive setting for the objective of minimizing the sum of completion times. This analysis is tight.
\end{theorem}

We only provide a sketch of the proof here, the full proof can be found in Appendix \ref{append:det_preemption}.

\begin{proof}[Proof sketch]
	We set $\alpha = \varphi$ and use Proposition \ref{prop:claims_case_dist}(b),(c) to bound the algorithmic running time $\pjA$ of a job $j$ by its optimal running time $\rho_j$.
	\begin{equation*}
		\pjA \le \varphi \rho_j.
	\end{equation*}
	
	We then compute the contribution of a job $k$ to a fixed job $j$ by grouping jobs based on their finishing order in the schedule. This allows us to estimate the completion time of job $j$:
	\begin{equation*}
		C_j \le \sum_{k > j} \pkA + j \cdot \pjA
	\end{equation*}
	Finally, we sum over all jobs to receive $\ALG \le 2\varphi \cdot \OPT$.
	
	To show that the analysis is tight, we provide an example where the algorithmic solution has a value of $2\varphi \cdot \OPT$ if we let the number of jobs approach infinity.
\end{proof}

The following theorem additionally provides a lower bound for the deterministic preemptive setting, giving us a first simple lower bound for this case. The proof is based on the lower bound provided in \cite{Duerr2018} for the deterministic non-preemptive case. We refer to Appendix \ref{append:lower_bound_preempt} for this proof.

\begin{theorem}
\label{thm:lower_bound_preempt}
	No algorithm in the preemptive deterministic setting can be better than $1.8546$-competitive.
\end{theorem}

\section{Randomized Setting}
\label{sec:randomized}

In this section we introduce randomness to further improve the competitive ratio of Algorithm \ref{alg:alpha_beta_sort}. There are two natural places to randomize: when deciding which jobs to test and the decision about the ordering of the jobs. These decisions directly correspond to the parameters $\alpha$ and $\beta$.

Making $\alpha$ randomized, for instance, could be achieved by defining $\alpha$ as a random variable with density function $f_\alpha: [1,\infty] \rightarrow \R_0^+$ and testing $j$ if and only if $r_j := u_j/t_j \ge \alpha$. Then the probability for testing $j$ would be given by $p = \int_1^{r_j} f_\alpha(x) dx$. Using a random variable $\alpha$ like this would make the analysis unnecessarily complicated, therefore we directly consider the probability $p$ without defining a density, and let $p$ depend on $r_j$. This additionally allows us to compute the probability of testing \emph{independently} for each job.

Introducing randomness for $\beta$ is even harder. The choice of $\beta$ influences multiple jobs at the same time, therefore independence is hard to establish. Additionally, $\beta$ appears in the denominator of our analysis frequently, hindering computations using expected values. We therefore forgo using randomness for the $\beta$-parameter and focus on $\alpha$ in this paper. We encourage future research to try their hand at making $\beta$ random.

We give a short pseudo-code of our randomized algorithm in Algorithm \ref{alg:rand_sort}. It is given a parameter-function $p(r_j)$ and a parameter $\beta$, both of which are to be determined later.

\begin{algorithm}[ht]
	$T \leftarrow \emptyset$, $N \leftarrow \emptyset$, $\sigma_j \equiv 0$\;
	\ForEach{$j\in[m]$}{
		add $j$ to $T$ with probability $p(r_j)$ and set $\sigma_j \leftarrow \beta t_j$\;
		otherwise add it to $N$ and set $\sigma_j \leftarrow u_j$\;
	}
	\While{$N \cup T \neq \emptyset$}{
		choose $j_{\min} \in \argmin_{j \in N \cup T} \sigma_j$\;
		\uIf{$j_{\min} \in N$}{execute $j_{\min}$ untested\; remove $j_{\min}$ from $N$\;}
		\uElseIf{$j_{\min} \in T$}{
			\eIf{$j_{\min}$ not tested}
			{test $j_{\min}$\; set $\sigma_{j_{\min}} \leftarrow p_{j_{\min}}$\;}
			{execute $j_{\min}$\; remove $j_{\min}$ from $T$\;}
		}
	}
	\caption{Randomized-SORT}
	\label{alg:rand_sort}
\end{algorithm}

\begin{theorem}
\label{thm:rand_sort}
	Randomized-SORT is $3.3794$-competitive for the objective of minimizing the sum of completion times.
\end{theorem}

\begin{proof}
	Again, we let $\rho_1 \ge \hdots \ge \rho_n$ denote the ordered optimal running time of jobs $1,\dots,n$. The optimal objective value is given by \eqref{eq:opt_value}. Fix jobs $j$ and $k$. For easier readability, we write $p$ instead of $p(r_j)$. Since the testing decision is now done randomly, the algorithmic running time $\pjA$ as well as the contribution $c(k,j)$ are now random variables. It holds 
	\begin{equation*}
		\pjA = 	\begin{cases}
					t_j + p_j 	&\text{with probability } p\\
					u_j 		&\text{with probability } 1-p\\
				\end{cases}
	\end{equation*}
	
	For the values of $c(k,j)$ we consult the case distinctions from the proof of the Contribution Lemma \ref{lem:contribution_lemma}. If $j \in N$, one can easily determine that $c(k,j) \le (1+1/\beta)u_j$ for all cases. Note that for this we did not need to use the final estimates with parameter $\alpha$ from the case distinction. Therefore this upper bound holds deterministically as long as we assume $j \in N$. By extension it also trivially holds for the expectation of $c(k,j)$:
	\begin{equation*}
		E[c(k,j)\,|\, j \text{ untested}]  \le (1+1/\beta)u_j.
	\end{equation*}
	Doing the same for the case distinction of $j \in T$, we get 
	\begin{equation*}
		E[c(k,j)\,|\, j \text{ tested}] \le \max\left((1+\beta)t_j,\left(1+\frac{1}{\beta}\right)p_j,t_j+p_j\right).
	\end{equation*}

	For the expected value of the contribution we have by the law of total expectation:
	\begin{equation*}
		\begin{aligned}
			E[c(k,j)] 	&= E[c(k,j)\,|\, j \text{ untested}] \cdot Pr[j \text{ untested}] \\
						&\qquad + E[c(k,j)\,|\, j \text{ tested}] \cdot Pr[j \text{ tested}]\\
						&\le \left(1+\frac{1}{\beta}\right) u_j \cdot (1-p) + \max\left((1+\beta)t_j,\left(1+\frac{1}{\beta}\right)p_j,t_j+p_j\right) \cdot p\\
		\end{aligned}
	\end{equation*}
	Note that this estimation of the expected value is independent of any parameters of $k$. That means, for fixed $j$ we estimate the contribution to be the same for all jobs with small parameter $k \le j$. Of course, as before, for the jobs with large parameter $k > j$ we may also alternatively directly use the algorithmic runtime of $k$:
	\begin{equation*}
		E[c(k,j)] \le E[\pkA].
	\end{equation*}
	
	Putting the above arguments together, we use the Contribution Lemma and linearity of expectation to estimate the completion time of $j$:
	\begin{equation*}
		\begin{aligned}
			E[C_j] 	&= \sum_{j=1}^n E[c(k,j)]\\
					&\le \sum_{k>j} E[\pkA] + \sum_{k \le j} E[c(k,j)].\\
		\end{aligned}
	\end{equation*}
	For the total objective value of the algorithm we receive again using linearity of expectation:
	\begin{equation*}
		\begin{aligned}
			E\left[\sum_{j=1}^n C_j\right] 	&\le \sum_{j=1}^n (j-1) E[\pjA] + \sum_{j=1}^n j \cdot E[c(k,j)]\\
											&\le \sum_{j=1}^n (j-1) (u_j \cdot (1-p) + (t_j+p_j) \cdot p)\\
											&\qquad + \sum_{j=1}^n j \Bigg( \left(1+\frac{1}{\beta}\right) u_j \cdot (1-p)\\
											&\qquad + \max\left((1+\beta)t_j,\left(1+\frac{1}{\beta}\right)p_j,t_j+p_j\right) \cdot p \Bigg)\\
											&\le \sum_{j=1}^n j \cdot \lambda_j(\beta,p),\\
		\end{aligned}
	\end{equation*}
	where we define
	\begin{equation*}
		\begin{aligned}
			\lambda_j(\beta,p) 	&:= \left(u_j+\left(1+\frac{1}{\beta}\right)u_j\right) \cdot (1-p)\\
								&\qquad + \left(t_j+p_j + \max\left((1+\beta)t_j,\left(1+\frac{1}{\beta}\right)p_j,t_j+p_j\right)\right) \cdot p.\\
		\end{aligned}
	\end{equation*}
	
	Having computed this first estimation for the objective of the algorithm, we now consider the ratio $\lambda_j(\beta,p) / \rho_j$ as a standalone. If we can prove an upper bound for this ratio, the same holds as competitive ratio for our algorithm.
	
	Hence the goal is to choose parameters $\beta$ and $p$, where $p$ can depend on $j$, \st $\lambda_j(\beta,p) / \rho_j$ is as small as possible. In the best case, we want to compute
	\begin{equation*}
		\min_{\beta \ge 1, p \in [0,1]} \ \max_j \ \frac{\lambda_j (\beta, p)}{\rho_j}.
	\end{equation*}
	
	\begin{lemma}
	\label{lem:min_max_rand}
		There exist parameters $\hat{\beta} \ge 1$ and $\hat{p} \in [0,1]$ s.t.
		\begin{equation*}
			\max_j \frac{\lambda_j (\hat{\beta}, \hat{p})}{\rho_j} \le 3.3794.
		\end{equation*}
	\end{lemma}

	The choice of parameters is given in the proof of the lemma, which can be found in Appendix \ref{append:min_max_rand}. During the proof we use computer-aided computations with Mathematica. The Mathematica code can be found in Appendix \ref{append:mathematica_code} and additionally on the webpage \cite{Albers2020} for download.
	
	To conclude the proof of the theorem, we write
	\begin{equation*}
		E\left[\sum_{j=1}^n C_j\right] \le \sum_{j=1}^n j \cdot \lambda_j (\hat{\beta}, \hat{p}) \le 3.3794 \sum_{j=1}^n j \cdot \rho_j =  3.3794 \cdot \OPT.
	\end{equation*}
\end{proof}

\section{Optimal Results for Minimizing the Makespan}
\label{sec:makespan}

In this section, we consider the objective of minimizing the makespan of our schedule. It turns out that we are able to prove the same tight algorithmic bounds for this objective function as D\"urr \etal in the unit-time testing case, both for deterministic and randomized algorithms. The decisions of the algorithms only depend on the ratio $r_j = u_j/t_j$. Refer to the appendix for the proofs.

\begin{theorem}
\label{thm:det_makespan}
	The algorithm that tests job $j$ iff $r_j \ge \varphi$ is $\varphi$-competitive for the objective of minimizing the makespan. No deterministic algorithm can achieve a smaller competitive ratio.
\end{theorem}

\begin{theorem}
\label{thm:rand_makespan}
	The randomized algorithm that tests job $j$ with probability $p = 1-1/(r_j^2-r_j+1)$ is $4/3$-competitive for the objective of minimizing the makespan. No randomized algorithm can achieve a smaller competitive ratio.
\end{theorem}

\section{Conclusion}
\label{sec:conclusion}

In this paper, we introduced the first algorithms for the problem of scheduling with testing on a single machine with general testing times that arises in the context of settings where a preliminary action can influence cost, duration or difficulty of a task. For the objective of minimizing the sum of completion times, we presented a $4$-approximation for the deterministic case, and a $3.3794$-approximation for the randomized case. If preemption is allowed, we can improve the deterministic result to $3.2361$. We also considered the objective of minimizing the makespan, for which we showed tight ratios of $1.618$ and $4/3$ for the deterministic and randomized cases, respectively.

Our results open promising avenues for future research, in particular tightening the gaps between our ratios and the lower bounds given by the unit case. Based on various experiments using different adversarial behaviour and multiple testing times it seems hard to force the algorithm to make mistakes that lead to worse ratios than those proven in \cite{Duerr2018} for the unit case. We conjecture that in order to achieve better lower bounds, the adversary must make live decisions based on previous choices of the algorithm, in particular depending on how much the algorithm has already tested, run or deferred jobs up to a certain point.

Further interesting directions for future work are the extension of the problem to multiple machines to consider scheduling problems like open shop, flow shop, or other parallel machine settings.

\bibliographystyle{splncs04}
\bibliography{SchedulingExplUncertain}

\clearpage
\appendix

\section{Jobs with Small Upper Limits}
\label{append:small_upper_limits}

In this section, we motivate our approaches by showing why the uniform algorithm cannot be simply extended to the general problem. An important insight of D\"urr \etal was that jobs with small upper bounds can be scheduled immediately without testing at the begin of any competitive algorithm. It turns out this does not generalize to non-uniform testing, which signifies that a new idea is necessary to deal with jobs that have small upper bounds when compared to their testing times. 

D\"urr \etal \cite{Duerr2018} proved for the special case of uniform testing times $t_j \equiv 1$ that \wloss any algorithm claiming competitive ratio $\lambda$ may start by scheduling all jobs $j$ with $u_j < \lambda$ untested in increasing order of $u_j$. Clearly, this statement does not hold directly for the general case, since a single job with $0 < u_j < \lambda$, $p_j = t_j = 0$ must be tested to yield a finite competitive ratio.

It seems intuitive to extend this idea to non-uniform testing by instead considering the ratio $\frac{u_j}{t_j}$ between upper bounds and test times. We show via a short counterexample that for any $\lambda \ge 1$ scheduling all jobs with $\frac{u_j}{t_j} < \lambda$ first leads to an arbitrarily bad result.

Consider the following instance: Given an integer $m$ and a small real number $\eps > 0$, we have $m$ jobs with $u_j = p_j = \lambda$ and $t_j = 1$ for $j = 1,\dots,m$. Clearly, all these jobs lie just over the limit $\lambda$ and are not considered for execution at the beginning of the schedule. Additionally we have a single extra job $j_0$ with $u_0 = m^2$, $t_0 = \frac{m^2}{\lambda} + \eps$, $p_0 = 0$.

An algorithm obeying the small upper limit rule schedules $j_0$ first, since we have $\frac{u_0}{t_0} < \lambda$ because of $\eps > 0$. Afterwards the remaining jobs are scheduled optimally, meaning  the algorithm runs them untested. Since the order of the remaining jobs is irrelavant, we may assume \wloss that the algorithm orders them by $1,\dots,m$.

For the completion time of $j_0$ we get $C_0 = m^2$ and for the other jobs we have $C_j = m^2 + j \cdot \lambda$. In total the value of the algorithm is:
\begin{equation*}
	\begin{aligned}
		\ALG	&= \sum_{j=0}^m C_j\\
				&= C_0 + \sum_{j=1}^m C_j\\
				&= m^2 + \sum_{j=1}^m (m^2 + j \cdot \lambda)\\
				&= m^3 + m^2\left(\frac{\lambda}{2} + 1\right) + m \frac{\lambda}{2}
	\end{aligned}
\end{equation*}

In contrast, the optimal offline schedule starts by running all jobs $j=1,\dots,m$ untested and then, leaving the large job $j_0$ for last, tests and runs it. We note here that the optimum only schedules $j_0$ last if $m^2/\lambda + \eps > \lambda$. We can guarantee this by choosing $m$ large enough.
\begin{equation*}
	\begin{aligned}
		\OPT	&= \sum_{j=0}^m C_j \\
				&= \sum_{j=1}^m C_j + C_0\\
				&= \sum_{j=1}^m (j \cdot \lambda) + \lambda m + \frac{m^2}{\lambda} + \eps + 0\\
				&= m^2\left(\frac{\lambda}{2} + \frac{1}{\lambda}\right) + m\frac{3\lambda}{2} + \eps
	\end{aligned}
\end{equation*}

Now, if we let $m \to \infty$ and $\eps \to 0$, it is clear that
\begin{equation*}
\frac{\ALG}{\OPT} \longrightarrow \infty.
\end{equation*}

The problem with scheduling $j_0$ first is that, while it might be reasonable to run it untested, all $m^2$ small jobs have to wait for it to finish before being scheduled, leading to a non-optimal result. Hence the problematic decision is the \emph{order} of the jobs rather than whether the algorithm tests or not.

The algorithm we propose in Section \ref{sec:deterministic} takes into account the ratio between upper bounds and test times, while additionally making sure that the execution length of both tested and untested jobs is considered in the order of the schedule.

\section{Lower Bounds for $(\alpha,\beta)$-SORT}

We give two concise lower bounds for the performance of $(\alpha,\beta)$-SORT.

First, we show that $(1,1)$-SORT cannot be better than at least $3$-competitive. Consider $n$ jobs with $u_j = p_j = 1$ and $t_j = 1-\eps$ for all jobs $j$. Since $u_j/t_j \ge 1$, $(1,1)$-SORT tests all jobs, and since $t_j < p_j$, it also runs all tests before executing anything. In total we have an algorithmic value of 
\begin{equation*}
\ALG = n^2 (1-\eps) + \frac{n^2}{2} + \frac{n}{2}.
\end{equation*}
Comparatively, the optimal runtime of a job is $\rho_j = 1$. This leads to a value of $\OPT = n^2/2 + n/2$. Therefore we get
\begin{equation*}
\frac{\ALG}{\OPT} = \frac{n^2 (1-\eps) + n^2/2 + n/2}{n^2/2 + n/2} \xrightarrow[\substack{n \to \infty\\ \eps \to 0}]{} 3.
\end{equation*}

Additionally, we can strengthen our analysis by proving a lower bound for the algorithm with general parameter choices $\alpha, \beta \ge 1$. We will show that $(\alpha,\beta)$-SORT is at most $2$-competitive for all such values.

Consider first the case $\alpha > 2$. Then the instance consisting of a single job with $u_j = 2, t_j = 1$ and $p_j = 0$ is executed untested by the algorithm, while the optimum tests the job. We have $\ALG/\OPT = 2$.

For the second case $1 \le \alpha \le 2$, we subdivide into two more cases based on the value of $\beta$. Let us start with $\beta > 2$. Consider for this $n$ jobs with $u_j = 2, t_j = 1, p_j = 2$. Since now $\alpha \le 2$, all jobs are tested by the algorithm. In particular, since $p_j = 2 > \beta = \beta t_j$, all jobs are tested before any executions happen. In total this gives an algorithmic value of 
\begin{equation*}
\ALG = n \cdot n + 2\left(\frac{n^2}{2}+\frac{n}{2}\right).
\end{equation*}

Similarly, the optimal value is $\OPT = 2 \cdot (n^2/2+n/2)$. We receive
\begin{equation*}
\frac{\ALG}{\OPT} = \frac{n^2 + 2 \cdot (n^2/2 + n/2)}{2 \cdot (n^2/2 + n/2)} \xrightarrow[n \to \infty]{} 2.
\end{equation*}

Finally, consider $1 \le \alpha \le 2$ and $\beta \ge 2$. We need two sets of jobs for this. First, we have a set of $n$ jobs with $u_j = \beta, t_j = 1-\eps, p_j = \beta$, similar to the previous instance. Since $u_j = \beta \ge \alpha > \alpha t_j$, these jobs are all tested by the algorithm. Since $p_j = \beta > \beta t_j$, all jobs are tested before any are executed.

Second, we have a set of $m$ jobs with $u_k = M, t_k = 1+\eps, p_k = 0$, where $M$ is some large number that does not play a role in either solutions. Since the upper bound is large, both the optimum and the algorithm test these jobs. Because $\beta t_k > p_j > \beta t_j$ for all jobs $k \in [m], j \in [n]$, we know that the algorithm sorts the executions as follows: First, all jobs in $[n]$ are tested. Then, all these jobs in $[n]$ are executed. Finally, all jobs in $[m]$ are tested and then immediately executed.

In total, this gives us an algorithmic value of
\begin{equation*}
\ALG = (1-\eps)n \cdot (n+m) + \beta \left(\frac{n^2}{2} + \frac{n}{2}\right) + \beta n \cdot m + (1+\eps) \left(\frac{m^2}{2}+\frac{m}{2}\right),
\end{equation*}
while the optimal value is
\begin{equation*}
\OPT = (1+\eps) \left(\frac{m^2}{2}+\frac{m}{2}\right) + (1+\eps) m \cdot n + \beta \left(\frac{n^2}{2} + \frac{n}{2}\right).
\end{equation*}

Now we choose $m = n$ and let $n \to \infty$ and $\eps \to 0$, giving us
\begin{equation*}
\frac{\ALG}{\OPT} \xrightarrow[\substack{n \to \infty \\ \eps \to 0}]{} \frac{3\beta + 5}{\beta + 3}.
\end{equation*}

This ratio is minimal for $\beta \ge 2$ at $\beta = 2$ with a value of $11/5 > 2$, finalizing the proof of the lower bound.

\section{A Simple Algorithm for the Uniform Case}

In the following we present a simpler 2-competitive algorithm for the unit-testing problem as compared to the Threshold algorithm from \cite{Duerr2018}. In the newest version of their complete paper \cite{Duerr2017}, D\"urr \etal add a note about this simpler algorithm, which they call DelayAll, and prove its competitive ratio. We present an alternative proof using our methods to highlight the differences between the proof techniques. The algorithm forces tests for all jobs designated for testing before executing any tested job. Since this leads to the same competitive ratio as the current best-known result, this represents an argument for increasing the relevance of testing further jobs as opposed to executing already tested jobs at any point in the schedule. Note that the optimal parameter choice of $\beta = 1$ in Algorithm \ref{alg:alpha_beta_sort} reflects this as well.

The algorithm first runs jobs with $u_j < 2$ untested, then tests all remaining jobs immediately. Finally, all jobs are run in order of shortest processing time first.

\begin{algorithm}[ht]
	$T \leftarrow \emptyset$, $N \leftarrow \emptyset$\;
	\ForEach{$j\in[m]$}{
		\eIf{$u_j \ge 2$}
		{add $j$ to $T$\;}
		{add $j$ to $N$\;}
	}
	run all jobs $j \in N$ untested\;
	test all jobs $j \in T$\;
	run all $j \in T$ in order of SPT\;
	\caption{Force Testing}
	\label{alg:force_testing}
\end{algorithm}

\begin{lemma}
	Force Testing is a 2-competitive algorithm for the objective of minimizing the sum of completion times in the unit-sized testing case.
\end{lemma}

\begin{proof}
	Again, we let $\rho_1 \ge \hdots \ge \rho_n$ denote the ordered optimal running time of jobs $1,\dots,n$. The optimal objective value is given by \eqref{eq:opt_value}. By Lemma 1 of \cite{Duerr2018} we may assume that all jobs fulfill $u_j \ge 2$ and are tested by the algorithm. We define $(\pi(j))_j$ to be the SPT order of the processing times $p_j$, such that $p_{\pi^{-1}(1)} \ge \dots \ge p_{\pi^{-1}(n)}$. This means the job with the largest processing time has index $1$ in this ordering, similar to the order of $\rho_j$. Then the value of the algorithm can be computed as follows: Since all jobs are tested first, an amount of $n \cdot n$ is added to the objective. Afterwards, the last job in the ordering $\pi$ (i.e.\ the job with the shortest processing time) contributes his processing time $n$ times. The second-to-last job contributes his $n-1$ times and so on. In total this gives
	\begin{equation*}
	\ALG = n^2 + \sum_{j \in [n]} \pi(j) \cdot p_j.
	\end{equation*}
	
	The SPT-rule is optimal on a single machine, see e.g.\ \cite{Pinedo2016}. Therefore if we re-order the jobs in the final step of the algorithm, the result can only get worse. We re-order according to the optimal order of $\rho_j$ and receive
	\begin{equation*}
	\sum_{j \in [n]} \pi(j) \cdot p_j \le \sum_{j \in [n]} j \cdot p_j.
	\end{equation*}
	Using $n^2 \le 2 \cdot \sum_{j=1}^n j$, we have
	\begin{equation*}
	\ALG \le \sum_{j \in [n]} j \cdot (2 + p_j).
	\end{equation*}
	
	If $\rho_j = 1+p_j$, then it holds that $2+p_j \le 2(1+p_j) = 2 \rho_j$. Similarly, if $\rho_j = u_j$, then by $u_j \ge 2$ for all jobs we have $2+p_j \le u_j+u_j = 2 \rho_j$. Inserting this into our estimation, we receive
	\begin{equation*}
	\ALG \le \sum_{j \in [n]} j \cdot 2 \rho_j = 2 \cdot \OPT
	\end{equation*}
\end{proof}

This analysis is tight, as can be seen by a simple example of $n$ jobs with large upper bounds and processing times $p_j = 0$. While the algorithm obliviously runs all tests first and has a completion time of $n$ for every job, the optimum tests and immediately runs every job, resulting in a value of $n^2/2 + n/2$. By this example we also see that while this oblivious algorithm has the same theoretical competitive ratio as the Threshold algorithm in \cite{Duerr2018}, there are instances where Threshold clearly performs better.

We also note that the idea of forcing tests cannot be extended to non-unit testing times. In the presence of arbitrarily large testing times, we may not prioritize these tests over potentially small execution times.

\section{Proofs}
\subsection{Proof of Proposition \ref{prop:claims_case_dist}}

\begin{proof}
	~
	\begin{enumerate}[label=(\alph*)]
		\item Let $j \in T$. Assume $\rho_j = t_j+p_j$. Then the result follows immediately from non-negativity of testing and running times. On the other hand, let $\rho_j = u_j$. Then, since $j \in T$ and $\alpha \ge 1$, we have $t_j \le \frac{1}{\alpha} u_j \le \rho_j$. By definition of the upper bound we also have $p_j \le \rho_j$.
		\item Let $j \in T$. By definition $u_j/t_j \ge \alpha$ and $\pjA = t_j + p_j$. If the optimum tests $j$, then $\rho_j = t_j+p_j$ and therefore $\pjA = \rho_j$. If, on the other hand, $\rho_j = u_j$, then:
		\begin{equation*}
		\frac{\pjA}{\rho_j} = \frac{t_j+p_j}{u_j} = \frac{t_j}{u_j} + \frac{p_j}{u_j} \le \frac{1}{\alpha} + 1
		\end{equation*}
		\item Now let $j \in N$ and hence $u_j/t_j < \alpha$ as well as $\pjA = u_j$. If the optimum doesn't test $j$, then $\rho_j = u_j$ and therefore $\pjA = \rho_j$. If it does, then $\rho_j = t_j+p_j$ and:
		\begin{equation*}
		\frac{\pjA}{\rho_j} = \frac{u_j}{t_j+p_j} \le \frac{u_j}{t_j} < \alpha
		\end{equation*}
	\end{enumerate}
\end{proof}

\subsection{Proof of Lemma \ref{lem:contribution_lemma}}
\label{append:contribution_lemma}

\begin{proof}
Let $j \in [n]$. By definition, the completion time $C_j$ is the point in time in the schedule when the entire execution of $j$, including a potential test in our case, is finished. Since our algorithm schedules all jobs on a single machine without waiting times, the completion time is equal to the amount of time the algorithm spends scheduling any job (including $j$ itself) before reaching this point. Hence, by the definition of the contribution:
\begin{equation*}
	C_j = \sum_{k \in [n]} c(k,j)
\end{equation*}
	
Now let additionally $k \in [n]$ and let $c(k,j)$ be the contribution of $k$ to the completion time of $j$. We do a rigorous case distinction over the possible values of $c(k,j)$, depending on the testing status of $j$ and $k$. Consider Figure \ref{fig:case_dist_untested} for an overview of the case distinction for $j \in N$ and Figure \ref{fig:case_dist_tested} for an overview for $j \in T$.

\begin{figure}[hbt]
	\centering
	\tikzstyle{main node}=	[draw,fill=TUMivory,ellipse,inner sep=4pt,very thick,minimum width=30mm,minimum height=10mm,font=\Large]
	\tikzstyle{end node}=	[draw,fill=TUMblue!25,ellipse,inner sep=4pt,minimum width=15mm,font=\Large]
	\tikzstyle{sloped node}=[above,sloped,font=\large]
	\resizebox{\linewidth}{!}{
	\begin{tikzpicture}[>=latex']
	\node[main node] (m1) at (0,0) {$k \in N$};
	\node[end node] (1) at (-1.5,-3) {$u_k$};
	\node[end node] (2) at (1.5,-3) {$0$};
	\draw[->,thick] (m1) -- node[style=sloped node]{$u_k \le u_j$} (1);
	\draw[->,thick] (m1) -- node[style=sloped node]{$u_k > u_j$} (2);
	\node [below right=0mm and 0mm of 1] {1.};
	\node [below right=0mm and 0mm of 2] {2.};
	\node[main node] (m2) at (8,0) {$k \in T$};
	\node[main node] (m3) at (7,-3) {$k$ tested};
	\node[end node] (3) at (6,-6) {$t_k+p_k$};
	\node[end node] (4) at (9,-6) {$t_k$};
	\node[end node] (5) at (10,-3) {$0$};
	\draw[->,thick] (m2) -- node[style=sloped node]{$\beta t_k \le u_j$} (m3);
	\draw[->,thick] (m2) -- node[style=sloped node]{$\beta t_k > u_j$} (5);
	\draw[->,thick] (m3) -- node[style=sloped node]{$p_k \le u_j$} (3);
	\draw[->,thick] (m3) -- node[style=sloped node]{$p_k > u_j$} (4);
	\node [below right=0mm and 0mm of 3] {3.};
	\node [below right=0mm and 0mm of 4] {4.};
	\node [below right=0mm and 0mm of 5] {5.};
	\end{tikzpicture}
	}
	\caption{Case distinction for contribution to job $j \in N$. The connecting lines are labeled with the relation between the involved parameters. The values with blue background correspond to $c(k,j)$ and are numbered from 1 to 5 distinct cases.}
	\label{fig:case_dist_untested}
\end{figure}
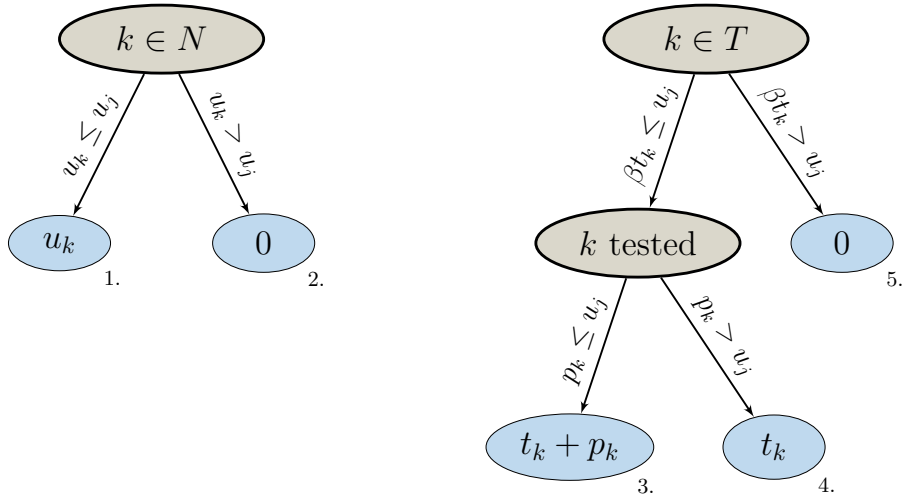

We start with $j \in N$. The cases are numbered exactly according to Figure \ref{fig:case_dist_untested}.
\begin{enumerate}
	\item $k \in N$, $u_k \le u_j$: The contribution is $c(k,j) = u_k \le u_j = \pjA \le \alpha \rho_j$ by Prop.~\ref{prop:claims_case_dist}(c).
	\item $k \in N$, $u_k > u_j$: $c(k,j) = 0$.
	\item $k \in T$, $\beta t_k \le u_j, p_k \le u_j$: $c(k,j) = t_k+p_k \le \left(1+\frac{1}{\beta}\right) u_j \le \left(1+\frac{1}{\beta}\right) \alpha \rho_j$ by Prop.~\ref{prop:claims_case_dist}(c).
	\item $k \in T$, $\beta t_k \le u_j, p_k > u_j$: $c(k,j) = t_k \le \frac{1}{\beta} u_j \le \frac{1}{\beta} \alpha \rho_j$ by Prop.~\ref{prop:claims_case_dist}(c).
	\item $k \in T$, $\beta t_k > u_j$: $c(k,j) = 0$.
\end{enumerate}
By taking the maximum over the above cases, we know that for $j \in N$ and any $k \in [n]$:
\begin{equation}
\label{eq:contr_bd_untested}
c(k,j) \le \left(1+\frac{1}{\beta}\right) \alpha \rho_j.
\end{equation}

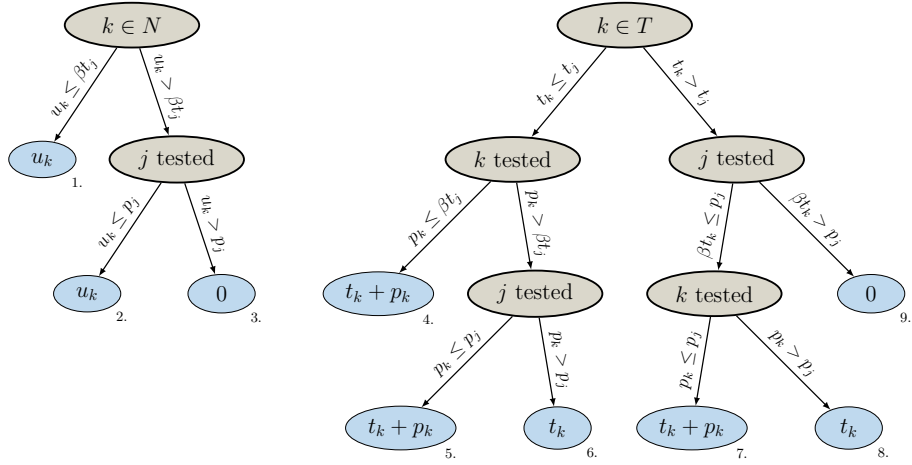
\begin{figure}[hbt]
	\centering
	\tikzstyle{main node}=	[draw,fill=TUMivory,ellipse,inner sep=4pt,very thick,minimum width=30mm,minimum height=10mm,font=\Large]
	\tikzstyle{end node}=	[draw,fill=TUMblue!25,ellipse,inner sep=4pt,minimum width=15mm,font=\Large]
	\tikzstyle{sloped node}=[above,sloped,font=\large]
	\resizebox{\linewidth}{!}{
		\begin{tikzpicture}[>=latex]
		\node[main node] (m1) at (1,0) {$k \in N$};
		\node[main node] (m2) at (2,-3) {$j$ tested};
		\node[end node] (1) at (-1,-3) {$u_k$};
		\node[end node] (2) at (0,-6) {$u_k$};
		\node[end node] (3) at (3,-6) {$0$};
		\draw[->,thick] (m1) -- node[style=sloped node]{$u_k \le \beta t_j$} (1);
		\draw[->,thick] (m1) -- node[style=sloped node]{$u_k > \beta t_j$} (m2);
		\draw[->,thick] (m2) -- node[style=sloped node]{$u_k \le p_j$} (2);
		\draw[->,thick] (m2) -- node[style=sloped node]{$u_k > p_j$} (3);
		\node [below right=0mm and 0mm of 1] {1.};
		\node [below right=0mm and 0mm of 2] {2.};
		\node [below right=0mm and 0mm of 3] {3.};
		\node[main node] (m3) at (12,0) {$k \in T$};
		\node[main node] (m4) at (9.5,-3) {$k$ tested};
		\node[main node] (m5) at (14.5,-3) {$j$ tested};
		\node[main node] (m6) at (10,-6) {$j$ tested};
		\node[main node] (m7) at (14,-6) {$k$ tested};
		\node[end node] (4) at (6.5,-6) {$t_k+p_k$};
		\node[end node] (5) at (7,-9) {$t_k+p_k$};
		\node[end node] (6) at (10.5,-9) {$t_k$};
		\node[end node] (7) at (13.5,-9) {$t_k+p_k$};
		\node[end node] (8) at (17,-9) {$t_k$};
		\node[end node] (9) at (17.5,-6) {$0$};
		\draw[->,thick] (m3) -- node[style=sloped node]{$t_k \le t_j$} (m4);
		\draw[->,thick] (m3) -- node[style=sloped node]{$t_k > t_j$} (m5);
		\draw[->,thick] (m4) -- node[style=sloped node]{$p_k \le \beta t_j$} (4);
		\draw[->,thick] (m4) -- node[style=sloped node]{$p_k > \beta t_j$} (m6);
		\draw[->,thick] (m5) -- node[style=sloped node]{$\beta t_k \le p_j$} (m7);
		\draw[->,thick] (m5) -- node[style=sloped node]{$\beta t_k > p_j$} (9);
		\draw[->,thick] (m6) -- node[style=sloped node]{$p_k \le p_j$} (5);
		\draw[->,thick] (m6) -- node[style=sloped node]{$p_k > p_j$} (6);
		\draw[->,thick] (m7) -- node[style=sloped node]{$p_k \le p_j$} (7);
		\draw[->,thick] (m7) -- node[style=sloped node]{$p_k > p_j$} (8);
		\node [below right=0mm and 0mm of 4] {4.};
		\node [below right=0mm and 0mm of 5] {5.};
		\node [below right=0mm and 0mm of 6] {6.};
		\node [below right=0mm and 0mm of 7] {7.};
		\node [below right=0mm and 0mm of 8] {8.};
		\node [below right=0mm and 0mm of 9] {9.};
		\end{tikzpicture}
	}
	\caption{Case distinction for contribution to job $j \in T$. The connecting lines are labeled with the relation between the involved parameters. The values with blue background correspond to $c(k,j)$ and are numbered from 1 to 9 distinct cases.}
	\label{fig:case_dist_tested}
\end{figure}

Now consider $j \in T$, the cases are numbered as seen in Figure \ref{fig:case_dist_tested}.
\begin{enumerate}
	\item $k \in N$, $u_k \le \beta t_j$: The contribution is $c(k,j) = u_k \le \beta t_j \le \beta \rho_j$ by Prop.~\ref{prop:claims_case_dist}(a).
	\item $k \in N$, $u_k > \beta t_j, u_k \le p_j$: $c(k,j) = u_k \le p_j \le \rho_j$ by Prop.~\ref{prop:claims_case_dist}(a).
	\item $k \in N$, $u_k > \beta t_j, u_k > p_j$: $c(k,j) = 0$.
	\item $k \in T$, $t_k \le t_j, p_k \le \beta t_j$: $c(k,j) = t_k+p_k \le (1+\beta) t_j \le (1+\beta) \rho_j$ by Prop.~\ref{prop:claims_case_dist}(a).
	\item $k \in T$, $t_k \le t_j, p_k > \beta t_j, p_k \le p_j$: $c(k,j) = t_k+p_k \le t_j+p_j = \pjA \le \left(1+\frac{1}{\alpha}\right) \rho_j$ by Prop.~\ref{prop:claims_case_dist}(b).
	\item $k \in T$, $t_k \le t_j, p_k > \beta t_j, p_k > p_j$: $c(k,j) = t_k \le t_j \le \rho_j$ by Prop.~\ref{prop:claims_case_dist}(a).
	\item $k \in T$, $t_k > t_j, \beta t_k \le p_j, p_k \le p_j$: $c(k,j) = t_k+p_k \le \left(1+\frac{1}{\beta}\right) p_j \le \left(1+\frac{1}{\beta}\right) \rho_j$ by Prop.~\ref{prop:claims_case_dist}(a).
	\item $k \in T$, $t_k > t_j, \beta t_k \le p_j, p_k > p_j$: $c(k,j) = t_k \le \frac{1}{\beta} p_j \le \frac{1}{\beta} \rho_j$ by Prop.~\ref{prop:claims_case_dist}(a).
	\item $k \in T$, $t_j > t_j, \beta t_k > p_j$: $c(k,j) = 0$.
\end{enumerate}
We again take the maximum over all cases, which yields that for $j \in T$ and any $k \in [n]$:
\begin{equation}
\label{eq:contr_bd_tested}
c(k,j) \le \max\left(1+\frac{1}{\alpha},1+\beta,1+\frac{1}{\beta}\right) \rho_j
\end{equation}

Combining equations \eqref{eq:contr_bd_untested} and \eqref{eq:contr_bd_tested} we achieve our desired bound.
\begin{equation*}
	\begin{aligned}
		c(k,j) 	&\le \max\left(\left(1+\frac{1}{\beta}\right) \alpha,1+\frac{1}{\alpha},1+\beta,1+\frac{1}{\beta}\right) \rho_j \\
				&= \max\left(\left(1+\frac{1}{\beta}\right) \alpha,1+\frac{1}{\alpha},1+\beta\right) \rho_j
	\end{aligned}
\end{equation*}
\end{proof}

\subsection{Proof of Theorem \ref{thm:det_preempt}}
\label{append:det_preemption}

\begin{proof}
	As before, we let $\rho_1 \ge \hdots \ge \rho_n$ denote the ordered optimal running times of jobs $1,\dots,n$. The optimal objective value is given by \eqref{eq:opt_value}. The decision of which jobs to test is exactly the same in Golden Round Robin as in the original $(\alpha,\beta)$-SORT. By Proposition \ref{prop:claims_case_dist}(b),(c) the algorithmic running time \eqref{eq:alg_proc_time} of all jobs is bounded by
	\begin{equation*}
	\pjA \le \max\left(\alpha,1+\frac{1}{\alpha}\right) \rho_j.
	\end{equation*}
	Minimizing this upper bound makes it clear why we set $\alpha = \varphi$ in the Golden Round Robin algorithm. We have
	\begin{equation}
	\label{eq:golden_ratio_estimate_proof_GRR}
	\pjA \le \varphi \cdot \rho_j.
	\end{equation}
	
	Similar to the proof of $(1,1)$-SORT, we now distinguish between how much $k$ contributes to the completion time $C_j$ for any two jobs $k,j$.
	
	Contrary to before, the contribution now only depends on the algorithmic running times of $k$ and $j$: If $\pkA \ge \pjA$ then $j$ finishes first in the Round Robin Scheme and $k$ contributes exactly $\pjA$ to the completion time of $j$. Otherwise, if $\pkA < \pjA$ then $k$ is done first and contributes its entire running time $\pkA$. We define $J_j^1 := \{k \in [n]: \pkA \ge \pjA\}$ and $J_j^2 := \{k \in [n]: \pkA < \pjA\}$. Then:
	\begin{equation*}
	C_j = \sum_{k \in J_j^1} \pjA + \sum_{k \in J_j^2} \pkA
	\end{equation*}
	
	Just as before, we divide the jobs into 'good' ($k > j$) and 'bad' ($k \le j$) jobs. For the good jobs, we use $\pkA$ and for the bad jobs we instead use $\pjA$. We use the properties of the sets $J_j^1$ and $J_j^2$ to estimate as follows:
	\begin{equation*}
	\begin{aligned}
	C_j 	&= \sum_{\substack{k \in J_j^1\\k>j}} \pjA + \sum_{\substack{k \in J_j^1\\k \le j}} \pjA + \sum_{\substack{k \in J_j^2\\k>j}} \pkA + \sum_{\substack{k \in J_j^2\\k \le j}} \pkA\\
	&\le \sum_{\substack{k \in J_j^1\\k>j}} \pkA + \sum_{\substack{k \in J_j^1\\k \le j}} \pjA + \sum_{\substack{k \in J_j^2\\k>j}} \pkA + \sum_{\substack{k \in J_j^2\\k \le j}} \pjA\\
	&= \sum_{k > j} \pkA + \sum_{k \le j} \pjA\\
	&= \sum_{k > j} \pkA + j \cdot \pjA\\
	\end{aligned}
	\end{equation*}
	
	For the sum of completion times, we receive
	\begin{equation*}
	\begin{aligned}
	\sum_{j=1}^n C_j 	&\le \sum_{j=1}^n \sum_{k=j+1}^n \pkA + \sum_{j=1}^n j \cdot \pjA\\
	&= \sum_{j=1}^n (j-1) \pjA + \sum_{j=1}^n j \cdot \pjA\\
	&\le 2 \sum_{j=1}^n j \cdot \pjA\\
	&\le 2 \varphi \sum_{j=1}^n j \cdot \rho_j\\
	&= 2 \varphi \cdot \OPT,\\
	\end{aligned}
	\end{equation*}
	where the last inequality follows from \eqref{eq:golden_ratio_estimate_proof_GRR}.

	We also show that this analysis of the Golden Round Robin algorithm is tight. For this, consider $n$ jobs with $u_j = p_j = 1$ and $t_j = 1/\varphi$ for all jobs. Since $u_j/t_j \ge \varphi$, the algorithm tests all jobs and therefore runs a round robin scheme on $n$ jobs with runtime $\pjA = 1+1/\varphi = \varphi$. We receive $\ALG = n^2 \varphi$. On the other hand, OPT does not test anything and has a total value of $\OPT = n^2/2 + n/2$. The final ratio is then
	\begin{equation*}
	\frac{\ALG}{\OPT} = \frac{n^2 \varphi}{n^2/2+n/2} \xrightarrow[n \to \infty]{} 2 \varphi.
	\end{equation*}
\end{proof}

\subsection{Proof of Theorem \ref{thm:lower_bound_preempt}}
\label{append:lower_bound_preempt}

\begin{proof}
	We prove the theorem by reducing a worst case scenario in the preemptive setting to the worst case provided for the non-preemptive setting in \cite{Duerr2018}. More specifically, we prove that given a preemptive algorithm for the adversarial scenario as defined in Chapter 3.2 of \cite{Duerr2018}, there exists a non-preemptive algorithm with competitive ratio at least as good as the given preemptive algorithm. Similarly, the optimal offline algorithm is always non-preemptive.
	
	We are given an instance with $t_j = 1$ and $u_j = \ubar \ge 1$ for all jobs. The adversarial strategy in \cite{Duerr2018} fixes a parameter $\delta \in [0,1]$ and then decides the runtime $p_j$ of all jobs as follows: 
	\begin{itemize}
		\item If a job is executed untested by the algorithm, set $p_j = 0$.
		\item If a job is tested by the algorithm, set $p_j = u_j$.
		\item If the number of jobs already decided is larger than $\delta n$, then always set $p_j = 0$.
	\end{itemize}
	This strategy can be easily extended to the preemptive case, where we just have to make sure that once an algorithm decides to test a job, it may not retract this decision later in order to deceive the adversary. If an algorithm starts testing a job (or execute it untested) even for a very small amount of time, it must continue to abide by this decision.
	
	Assume now that we are given a schedule produced by an algorithm $\ALG^{pre}$ which may be preemptive. We know that all jobs run untested by $\ALG^{pre}$ have processing time $p_j = 0$ as well as all jobs that have been decided after more than $\delta n$ other jobs have already been fixed. The rest of the jobs have processing time $p_j = u_j$.
	
	We fix an ordering of the \emph{execution instances} of the algorithm, which are defined as the exact points in time where the algorithm finishes either an untested execution, a test, or the execution of an already tested job.
	
	We then define a new algorithm $\ALG^*$, which will be non-preemptive, by the following rule: Go through the schedule of $\ALG^{pre}$ starting at time $0$. Whenever you encounter an execution instance in the schedule of $\ALG^{pre}$, schedule the corresponding execution or test completely without preemption in $\ALG^*$.
	
	By definition, the ordering of the execution instances stays the same from $\ALG^{pre}$ to $\ALG^*$. Therefore the completion time of any job can only get better, i.e.\ $C_j^* \le C_j^{pre}$. Additionally, the exact same set of jobs is tested as before and the set of $\delta n$ jobs that is decided first by the adversary is unchanged. Hence, the behavior of the adversary is the same for both algorithms and the optimal schedule does not change.
	
	Combining these arguments, we know that the competitive ratio of the non-preemptive algorithm can not be worse than that of the preemptive version:
	\begin{equation*}
	\frac{\ALG^*}{\OPT} \le \frac{\ALG^{pre}}{\OPT}
	\end{equation*}
	
	To complete the proof, we cite section 3.2 of \cite{Duerr2018}, where it is proven that for instances with $t_j = 1, u_j = \ubar \ge 1$ with an adversarial strategy defined as above there exists no deterministic algorithm with competitive ratio better than 1.8546. We have just proven that for all preemptive algorithms in this setting there always exists a non-preemptive algorithm with the same or better competitive ratio. Therefore this lower bound also holds for the preemptive version of the problem.
\end{proof}

\subsection{Proof of Lemma \ref{lem:min_max_rand}}
\label{append:min_max_rand}

\begin{proof}
Let $\lambda_j(\beta,p)$ be defined as in the proof of Theorem \ref{thm:rand_sort}. We want to find values $\hat{\beta}$ and $\hat{p}$, s.t.\ $\max_j \frac{\lambda_j (\hat{\beta}, \hat{p})}{\rho_j}$ is as small as possible.
	
Because $\rho_j = \min(u_j,t_j+p_j)$ is difficult to handle, we have to do a case distinction on its value. We start with $\rho_j = u_j$. Recall that $p_j \le u_j$ and that we defined $r = u_j/t_j$. For better readability, we drop the index $j$ of $r_j$ during this computation. We have 
\begin{equation*}
	\begin{aligned}
		\frac{\lambda_j(\beta,p)}{\rho_j}	&= \frac{u_j + \left(1+\frac{1}{\beta}\right) u_j}{u_j} \ (1-p)\\
											&\qquad + \frac{t_j+p_j + \max\left( (1+\beta) t_j, \left(1+\frac{1}{\beta}\right) p_j, t_j+p_j \right)}{u_j} \ p \\
											&\le \left(1 + 1 + \frac{1}{\beta}\right) (1-p) + \left( \frac{1}{r} + 1 + \max\left( (1+\beta) \frac{1}{r}, 1 + \frac{1}{\beta}, \frac{1}{r} + 1 \right) \right) p \\
											&= \left( \frac{1}{r} - 1 - \frac{1}{\beta} + \max\left( (1+\beta)\frac{1}{r}, 1 + \frac{1}{\beta}, 1 + \frac{1}{r} \right) \right) p + 2 + \frac{1}{\beta},\\
	\end{aligned}
\end{equation*}
which is a linear function in $p$ with strictly positive slope. Now let us look at $\rho_j = t_j+p_j$. We utilize $p_j \ge 0$ and $t_j \ge 0$ and write
\begin{equation*}
	\begin{aligned}
		\frac{\lambda_j(\beta,p)}{\rho_j}	&= \frac{u_j + \left(1+\frac{1}{\beta}\right) u_j}{t_j+p_j} \ (1-p)\\
											&\qquad + \frac{t_j+p_j + \max\left( (1+\beta) t_j, \left(1+\frac{1}{\beta}\right) p_j, t_j+p_j \right)}{t_j+p_j} \ p \\
											&\le \left(r + \left(1+\frac{1}{\beta}\right) r \right) (1-p) + \left( 1 + \max\left( 1+\beta, 1+\frac{1}{\beta}, 1 \right) \right) p \\
											&= \left(2 + \beta - \left(2 + \frac{1}{\beta}\right) r \right) p + \left(2 + \frac{1}{\beta}\right) r, \\
	\end{aligned}
\end{equation*}
where we used $\beta \ge 1$ in the last step. This is a linear function in $p$ with strictly negative slope if $r > \frac{2+\beta}{2+1/\beta}$. We observe that $r$ is the only parameter dependent on $j$ in either of these formulas. Therefore taking the maximum over all $j$ as required is equivalent to taking the maximum over values of $r \ge 1$.

We are interested in minimizing the maximum of these two functions over $0 \le p \le 1$. The minimal maximum of two linear function where one is strictly increasing and the other strictly decreasing is attained at their intersection. In the cases where the second function is in fact increasing, the minimal maximum may be attained somewhere else and have a smaller value, but we can ignore this and still use the intersection (as long as the intersection point lies in $[0,1]$), because this can only provide a worse upper bound.

We use the mathematical solver Mathematica for some of the following computations. Please consult Appendix \ref{append:mathematica_code} or webpage \cite{Albers2020} for the complete Mathematica code for this section. Using this code, we compute the intersection and receive the following value for $p$:
\begin{equation}
\label{eq:p_solution}
p(r) = \frac{r^2 + 2 \beta r^2 - r - 2 \beta r }{r^2 + 2 \beta r^2 - r - 3 \beta r - \beta^2 r + \beta + \beta r \max\left( (1+\beta)\frac{1}{r}, 1 + \frac{1}{\beta}, 1 + \frac{1}{r} \right)}
\end{equation}
For most feasible values of $\beta$ and $r$, this fraction lies between $[0,1]$, but not for all. We will have to make sure that our final choice of $p$ fulfills this requirement.

We insert this value for $p$ into either of the above linear functions. Afterwards the result depends only on $r$ and $\beta$. Since $r$ is our job parameter, we must consider the worst case (i.e.\ the maximum case) in dependence of $r \ge 1$. Therefore we run a parameter search for $\beta$, such that this worst-case value is minimized. See the Mathematica code for the exact computation. The result of the search was $\beta \approx 1.2574$ with a value of approximately $3.3794$ and a maximizing value of $r \approx 1.4386$.

We therefore want to choose $\beta = \hat{\beta} := 1.2574$ and $p(r)$ as in equation \eqref{eq:p_solution}. For this value of $\beta$, the definition of $p(r)$ is non-negative for all $r \ge 1$. However, for some values of $r$ the definition is larger than $1$. This is obviously not admissible and therefore we choose
\begin{equation*}
\hat{p}(r) := \min(p(r),1).
\end{equation*}
Consult Figure \ref{fig:graph_p_phat} for an illustration of this definition.

\begin{figure}[tbh]
	\centering
	\includegraphics{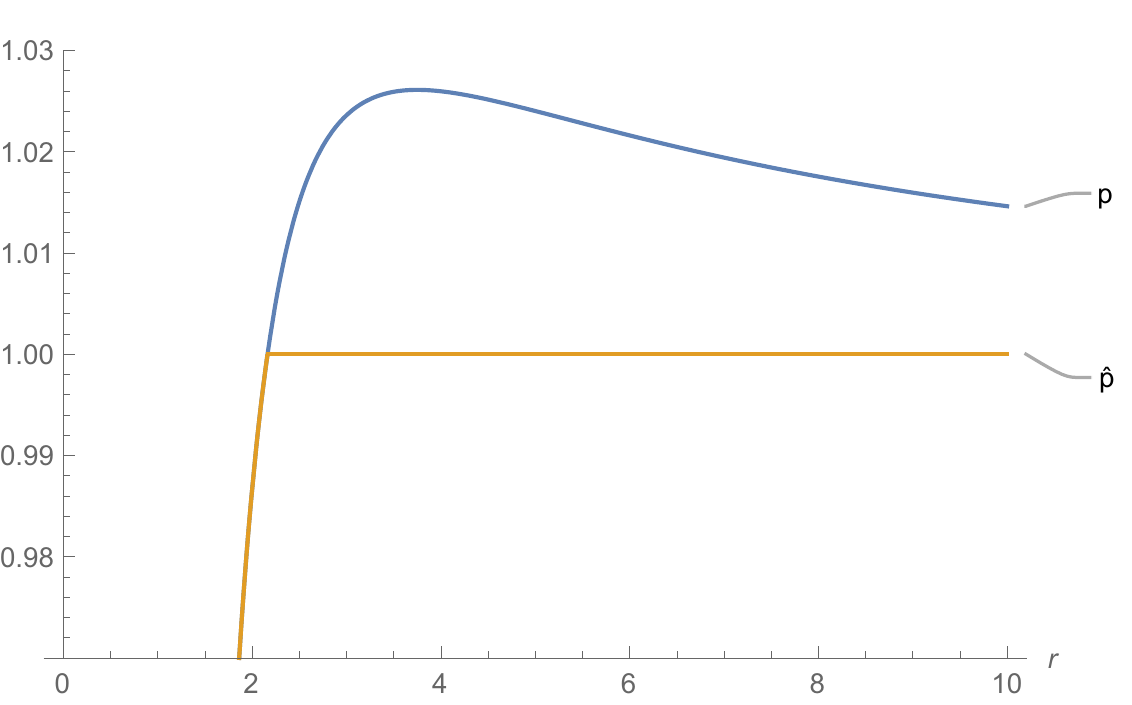}
	\caption{Graphs of $p$ and $\hat{p}$. For values larger than $\hat{r} \approx 2.1637$ the value of $p$ exceeds $1$.}
	\label{fig:graph_p_phat}
\end{figure}

Since we now restricted the choice of $\hat{p}$ for some values of $r$ to $1$, we need to make sure that the maximum value of our two functions for $p=1$ in these cases is smaller than the upper bound already provided. Otherwise our worst case estimate would increase. Using Mathematica, we determine that the value of $p(r)$ is only larger than 1 for values $r > \hat{r} \approx 2.1637$. We again consult Mathematica to compute
\begin{equation*}
\max_{r > \hat{r}} \frac{\lambda_j(\hat{\beta},1)}{\rho_j} \lessapprox 3.2574,
\end{equation*}
which is smaller than $3.3794$. Therefore
\begin{equation*}
	\begin{aligned}
		\max_{r \ge 1} \frac{\lambda_j(\hat{\beta},\hat{p}(r))}{\rho_j} &= \max\left( \max_{1 \le r \le \bar{r}} \frac{\lambda_j(\hat{\beta},p(r))}{\rho_j}, \max_{r > \hat{r}} \frac{\lambda_j(\hat{\beta},1)}{\rho_j} \right)\\
																		&\lessapprox \max(3.3794, 3.2574)\\
																		&= 3.3794.\\
	\end{aligned}
\end{equation*}

This concludes the proof.
\end{proof}

\subsection{Proof of Theorem \ref{thm:det_makespan}}

\begin{proof}
	Since every job contributes the same amount to the objective regardless of where in the schedule it is placed, we can assume worst case instances consist of only a single job. This statement is formalized in Lemma 20 of the full version of D\"urr \etal \cite{Duerr2017}. We therefore consider a single job j with upper bound $u_j$, processing time $p_j$ and testing time $t_j$.
	
	As seen in some of the previous proofs, a case distinction on the value of $\OPT = \rho_j = \min(u_j,t_j+p_j)$ is usually a good strategy. Therefore consider first $\OPT = u_j$. If the algorithm tests $j$, then by definition $\ALG = t_j+p_j$ as well as $r_j \ge \varphi$. Hence
	\begin{equation*}
	\frac{\ALG}{\OPT} = \frac{t_j+p_j}{u_j} \le \frac{1}{r_j} + 1 \le 1 + \frac{1}{\varphi} = \varphi,
	\end{equation*}
	where we used $p_j \le u_j$ and the defining property of the golden ratio. If on the other hand, the algorithm does not test $j$, then $\ALG = u_j = \OPT$.
	
	Now consider $\OPT = t_j+p_j$. If the algorithm tests $j$, then $\ALG = t_j+p_j = \OPT$. If it does not, then we have $\ALG = u_j$ as well as $r_j < \varphi$ and therefore
	\begin{equation*}
	\frac{\ALG}{\OPT} = \frac{u_j}{t_j+p_j} \le \frac{u_j}{t_j} = r_j \le \varphi,
	\end{equation*}
	where we used $p_j \ge 0$.
	
	It remains to show that no algorithm can achieve a better ratio. For this we use the proof from D\"urr \etal \cite{Duerr2018}, which we include for completeness. Consider an instance with a single job $j$ with $u_j = \varphi, t_j = 1$. An algorithm that tests the job has competitive ratio $\varphi$ if the adversary sets $p_j = 0$. An algorithm that doesn't test $j$ has competitive ratio $1+1/\varphi = \varphi$ for the case $p_j = u_j$.
\end{proof}

\subsection{Proof of Theorem \ref{thm:rand_makespan}}

\begin{proof}
	Consider a worst case instance consisting of a single job $j$ with upper bound $u_j$, processing time $p_j$, and testing time $t_j$. We can compute the expected value of the algorithmic solution:
	\begin{equation*}
	E[\ALG] = (t_j + p_j) \cdot p + u_j \cdot (1-p)
	\end{equation*}
	
	Again, we do a case distinction. If $\OPT = u_j$ then
	\begin{equation*}
	\frac{E[\ALG]}{\OPT} = \frac{(t_j + p_j) \cdot p + u_j \cdot (1-p)}{u_j} \le \left(\frac{1}{r_j} + 1\right) \cdot p + 1-p.
	\end{equation*}
	Otherwise, if $\OPT = t_j + p_j$ then
	\begin{equation*}
	\frac{E[\ALG]}{\OPT} = \frac{(t_j + p_j) \cdot p + u_j \cdot (1-p)}{t_j+p_j} \le p + r_j \cdot (1-p).
	\end{equation*}
	
	To achieve a good competitive ratio, we want to minimize the maximum of these two functions. We again do this by computing their intersection point as in the proof of Lemma \ref{lem:min_max_rand}. To simplify presentation, we only show that the resulting value for $p$ is indeed optimal: We insert $p = 1-1/(r_j^2-r_j+1)$ into both expressions and receive after a bit of algebra:
	\begin{equation*}
	\frac{E[\ALG]}{\OPT} \le \frac{r_j^2}{r_j^2-r_j+1}.
	\end{equation*}
	This function is maximized at $r_j = 2$ with value $4/3$.
	
	To proof optimality of this algorithm we provide a randomized instance where every deterministic algorithm is at most $4/3$-competitive. This proof is again adopted from D\"urr \etal \cite{Duerr2018}. Consider a single job instance with $u_j = 2, t_j = 1$ that has $p_j = 0$ and $p_j = 2$ both with probability $1/2$. Both the deterministic algorithm that tests the job and the one that doesn't test have expected makespan $2$. The optimal solution has an expected value of $3/2$. Therefore every deterministic algorithm is at most $4/3$-competitive. Applying Yao's principle \cite{Yao1977} gives the desired result.
\end{proof}

\section{Mathematica Code for the Randomized Algorithm}
\label{append:mathematica_code}
See appended pages.

\newpage


\section*{Supplementary material}

In[1]:= f1[beta_, r_, p_] := p (1/r - 1 - 1/beta + Max[(1 + beta) 1/r, 1 + 1/beta, 1 + 1/r]) + 2 + 1/beta;

In[2]:= f2[beta_, r_, p_] := p (2 + beta - (2 + 1/beta) r) + (2 + 1/beta) r;

In[3]:= (*test for positive slope of f1*)
Reduce[{1/r - 1 - 1/beta + Max[(1 + beta) 1/r, 1 + 1/beta, 1 + 1/r] > 0, beta ≥ 1, r ≥ 1}]

Out[3]= beta ≥ 1 && r ≥ 1

In[4]:= (*test for positive slope of f2*)
Reduce[{2 + beta - (2 + 1/beta) r < 0, beta ≥ 1, r ≥ 1}]

Out[4]= beta ≥ 1 && r > (2 beta + beta^2)/(1 + 2 beta)

In[5]:= (*find intersection of f1 and f2*)
Reduce[{f1[beta, r, p] == f2[beta, r, p], r ≥ 1, beta ≥ 1}]

Out[5]= ((r == 1 && beta ≥ 1) || (r > 1 && (1 ≤ beta < (2 - 3 r + 2 r^2)/(2 (-1 + r)) + 1/2 √((4 - 8 r + 9 r^2 - 8 r^3 + 4 r^4)/(-1 + r)^2) || beta > (2 - 3 r + 2 r^2)/(2 (-1 + r)) + 1/2 √((4 - 8 r + 9 r^2 - 8 r^3 + 4 r^4)/(-1 + r)^2)))) && p == (-r - 2 beta r + r^2 + 2 beta r^2)/(beta - r - 3 beta r - beta^2 r + r^2 + 2 beta r^2 + beta r Max[1 + 1/beta, 1 + 1/r, (1+beta)/r])

In[6]:= (*insert value for p into f1*)
fun[beta_ ?NumericQ] := Maximize[
{f1[beta, r, (-r - 2 beta r + r^2 + 2 beta r^2)/(beta - r - 3 beta r - beta^2 r + r^2 + 2 beta r^2 + beta r Max[1 + 1/beta, 1 + 1/r, (1+beta)/r])],
r ≥ 1}, r][[1]];

In[7]:= `Plot[{fun[beta]}, {beta, 0, 3}, PlotRange → {3, 8}]`

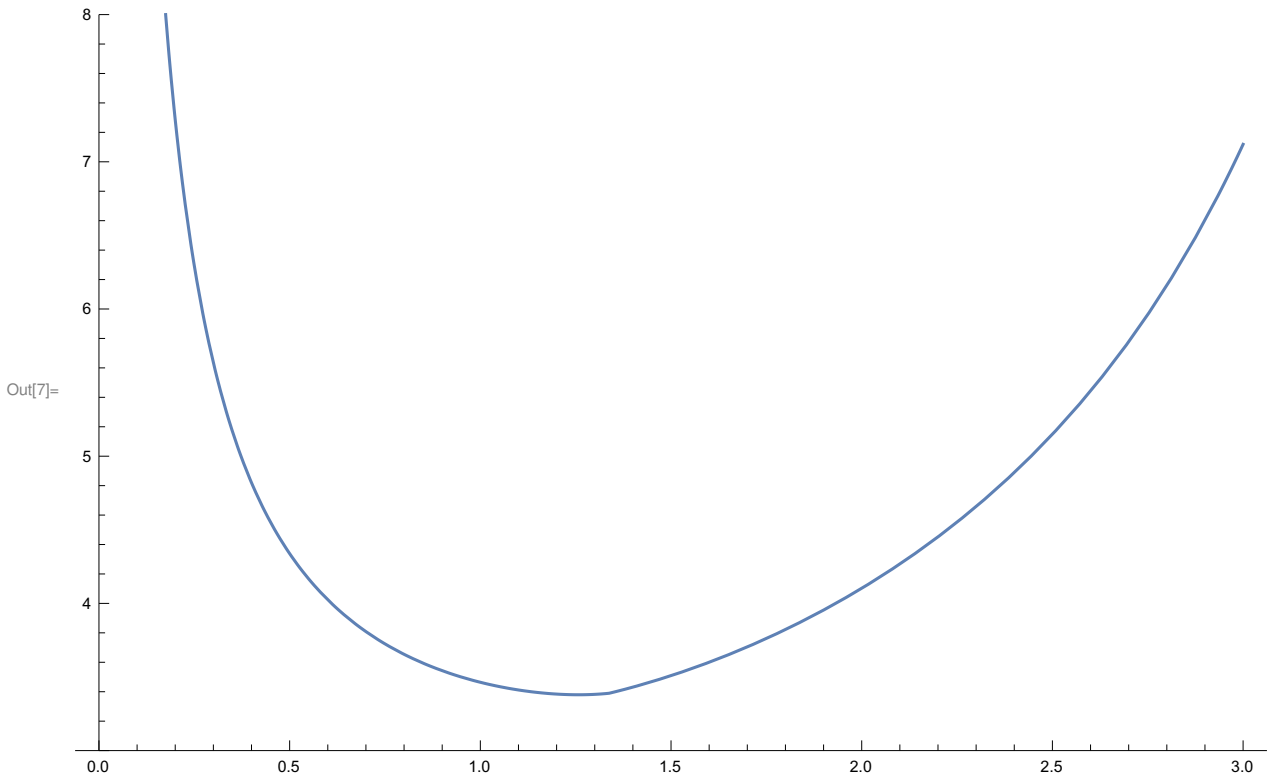

Out[7]=

In[8]:= `(*find beta≥1 such that fun[beta] is minimized*)`
`NMinimize[{fun[beta], 2 ≥ beta ≥ 1}, beta, MaxIterations → 200]`

Out[8]= `{3.3794, {beta → 1.25744}}`

In[9]:= `beta = 1.2574377980688796`; ratio = 3.3793988072916914`;`

In[10]:= `(*find maximizing value of r*)`
$$\text{Maximize}\left[\left\{f1\left[beta, r, \frac{-r - 2\,beta\,r + r^2 + 2\,beta\,r^2}{beta - r - 3\,beta\,r - beta^2\,r + r^2 + 2\,beta\,r^2 + beta\,r\,\text{Max}\left[1 + \frac{1}{beta}, 1 + \frac{1}{r}, \frac{1+beta}{r}\right]}\right], r \geq 1\right\}, r\right]$$

Out[10]= `{3.3794, {r → 1.43864}}`

In[11]:= `(*for completeness, also compute the worst case for f2*)`
$$\text{Maximize}\left[\left\{f2\left[beta, r, \frac{-r - 2\,beta\,r + r^2 + 2\,beta\,r^2}{beta - r - 3\,beta\,r - beta^2\,r + r^2 + 2\,beta\,r^2 + beta\,r\,\text{Max}\left[1 + \frac{1}{beta}, 1 + \frac{1}{r}, \frac{1+beta}{r}\right]}\right], r \geq 1\right\}, r\right]$$

Out[11]= `{3.3794, {r → 1.43864}}`

In[14]:= `(*insert beta into p*)`
pfun[r_] := $\dfrac{-r - 2\,\text{beta}\,r + r^2 + 2\,\text{beta}\,r^2}{\text{beta} - r - 3\,\text{beta}\,r - \text{beta}^2\,r + r^2 + 2\,\text{beta}\,r^2 + \text{beta}\,r\,\text{Max}\!\left[1 + \frac{1}{\text{beta}},\, 1 + \frac{1}{r},\, \frac{1+\text{beta}}{r}\right]}$;

In[15]:= `(*since p is >1 for some values, use phat instead*)`
phatfun[r_] := Min[pfun[r], 1];

In[16]:= Plot[{pfun[r], phatfun[r]}, {r, 1, 10},
  PlotRange → {0.97, 1.03}, AxesLabel → {r,}, PlotLabels -> {"p", "p̂ "}]

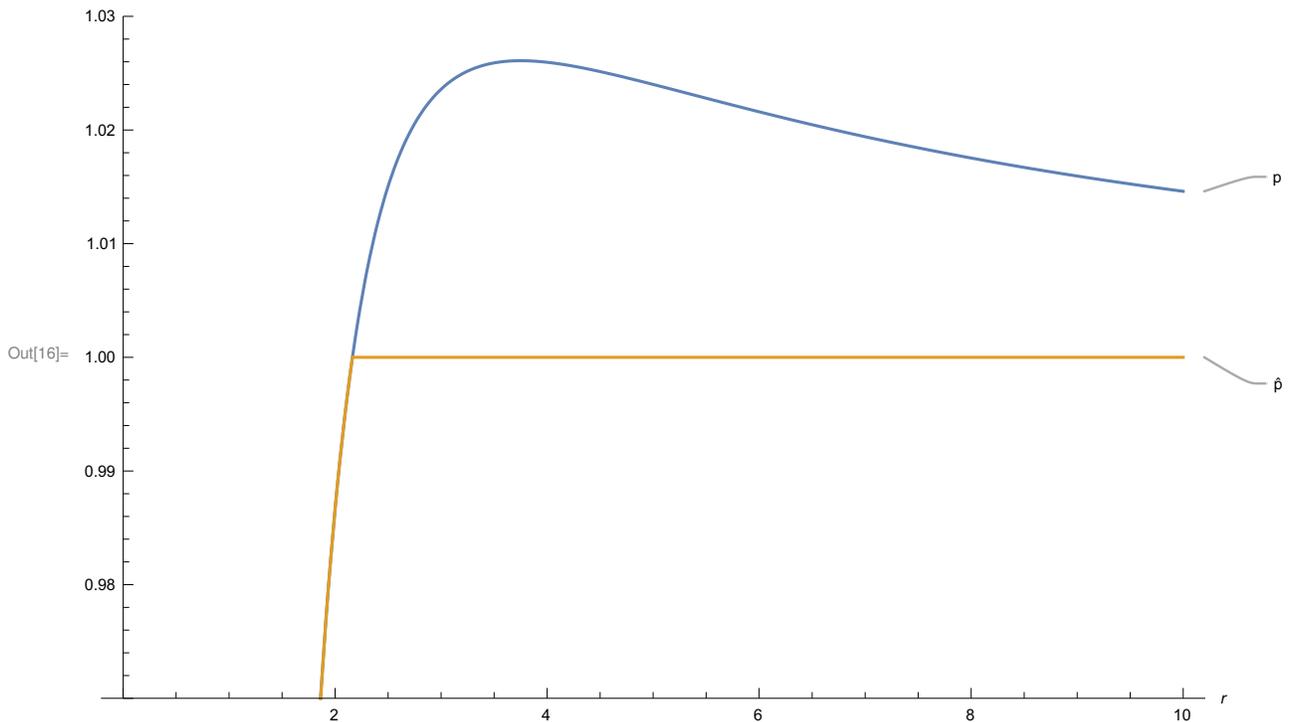

In[17]:= `(*compute interval where p<0*)`Reduce[pfun[r] < 0]

⋯ Reduce: Reduce was unable to solve the system with inexact coefficients. The answer was obtained by solving a corresponding exact system and numericizing the result.

Out[17]= 0 < r < 1.

In[18]:= `(*compute interval where p>1*)`Reduce[pfun[r] > 1]

⋯ Reduce: Reduce was unable to solve the system with inexact coefficients. The answer was obtained by solving a corresponding exact system and numericizing the result.

Out[18]= r > 2.16371

In[19]:= rhat = 2.163706796301143`;

In[20]:= `(*making sure the values of f1 for large r isn't larger than our current worst case*)`
`Maximize[{f1[beta, r, 1], r ≥ rhat}, r]`

Out[20]= {3.25744, {r → 2.16371}}

In[21]:= `(*making sure the values of f2 for large r isn't larger than our current worst case*)`
`Maximize[{f2[beta, r, 1], r ≥ rhat}, r]`

Out[21]= {3.25744, {r → 2.16371}}



\end{document}